\documentclass[oribibl]{llncs}
\usepackage{amsopn}
\usepackage{enumerate}
\usepackage{tikz}
\usetikzlibrary{trees}

\usepackage[utf8]{inputenc}
\usepackage{amsmath}
\usepackage{amssymb,amsfonts}
\usepackage[T1]{fontenc}

\usepackage{graphicx}
\usepackage{hyperref}
\usepackage{wrapfig}

\title{Characterizations and approximability of hard counting classes below \sP}
\author{Eleni Bakali \and Aggeliki Chalki \and Aris Pagourtzis}
\authorrunning{E. Bakali, A. Chalki, A. Pagourtzis}

\institute{School  of Electrical and Computer Engineering,\\
National Technical University of Athens, 15780 Athens, Greece\\
\email{mpakali@corelab.ntua.gr, achalki@corelab.ntua.gr, pagour@cs.ntua.gr}}

\renewenvironment{proof}{\medskip\noindent\textit{Proof}. }{\hfill $\Box$}

\newcommand{\totp}{\ensuremath{\mathsf{TotP}}}
\newcommand{\sP}{\ensuremath{\#\mathsf{P}}}
\newcommand{\sPE}{\ensuremath{\#\mathsf{PE}}}
\newcommand{\cP}{\ensuremath{\mathsf{P}}}

\newcommand{\NP}{\ensuremath{\mathsf{NP}}}
\newcommand{\RP}{\ensuremath{\mathsf{RP}}}
\newcommand{\RPo}{\ensuremath{\#\mathsf{RP}_1}}
\newcommand{\RPt}{\ensuremath{\#\mathsf{RP}_2}}
\newcommand{\FP}{\ensuremath{\mathsf{FP}}}
\newcommand{\FPRASP}{\ensuremath{\mathsf{FPRAS}'}}
\newcommand{\FPRAS}{\ensuremath{\mathsf{FPRAS}}}
\newcommand{\sBPP}{\ensuremath{\#\mathsf{BPP}}}
\newcommand{\BPP}{\ensuremath{\mathsf{BPP}}}

\newcommand{\sSAT}{\textsc{\#Sat}}

\newcommand{\dnf}{\textsc{\#Dnf}}

\begin{document}

\maketitle

\begin{abstract}
An important objective of research in counting complexity is to understand which counting problems are approximable. In this quest, the complexity class \totp, a hard subclass of \sP, is of key importance, as it contains self-reducible counting problems with easy decision version, thus eligible to be approximable. Indeed, most problems known so far to admit an fpras fall into this class.
  
An open question raised recently by the community of descriptive complexity is to find a logical characterization of \totp\ and of \emph{robust} subclasses of \totp.\ In this work we define two subclasses of \totp, in terms of descriptive complexity, both of  which are robust in the sense that they have natural complete  problems, which are defined in terms of satisfiability of Boolean formulae. 
  
We then explore the relationship between the class of approximable counting problems and \totp.
 We prove that $\totp\nsubseteq\FPRAS$ if and only if $\NP\neq\RP$ and $\FPRAS\nsubseteq\totp$ unless \RP=\cP. To this end we introduce two ancillary classes that can both be seen as counting versions of \RP. We further show that
 \FPRAS\ lies between one of these classes and a counting version of \BPP.
 
 Finally, we provide a complete picture of inclusions among all the classes defined or discussed in this paper with respect to  different conjectures about the \NP\ vs. \RP\ vs. \cP\ questions.

\end{abstract}
\section{Introduction}

The class \sP ~\cite{Valiant79} is the class of functions that count the number of solutions to problems in \NP,\ e.g. \textsc{\#Sat} is the function that on input a formula $\phi$ returns the number of satisfying assignments of $\phi.$ Equivalently, functions in \sP\ count accepting paths of non-deterministic polynomial time Turing machines (NPTMs).

 \NP-complete problems are hard to count, but it is not the case that problems in \cP\ are easy to count as well. When we consider counting, non-trivial facts hold. First of all there exist \sP-complete problems, that have decision version in \cP,\ e.g. \dnf. Moreover, some of them can be approximated, e.g. the Permanent~\cite{JS96perm} and \dnf ~\cite{KLM89}, while others cannot, e.g. \textsc{\#Is}~\cite{DFJ02}. 
  The class of problems in \sP\ with decision version in \cP\ is called \sPE, and a subclass of \sPE\ is \totp, which contains all self-reducible problems in \sPE ~\cite{PZ06}. Their significance will be apparent in what follows.
 
 Since many counting problems cannot be exactly computed in polynomial time, the interest of the community has  turned to the complexity of approximating them. On one side, there is an enormous literature on approximation algorithms and inapproximability results for individual problems in \sP ~\cite{DFJ02,GGJ17,JS96perm,KLM89,Valiant79}.  On the other hand, there have been attempts to classify counting problems with respect to their approximability~\cite{Arenas19,Arenas,DGGJ04, Saluja}.

\subsubsection{Related work.} From a unifying point of view, the most important results regarding approximability are the following. 
Every function in \sP\ either admits an fpras, or does not admit any polynomial approximation ratio~\cite{JS89}; we will therefore call the latter \emph{inapproximable}.
For self-reducible problems in \sP,\ fpras is equivalent to almost uniform sampling~\cite{JS89}.
With respect to approximation preserving reductions, there are three main classes of functions in \sP ~\cite{DGGJ04}: (a) the class of functions  that admit an fpras, (b) the class of functions that are interreducible with \sSAT, and  (c) the class $\#\mathsf{RH\Pi_1}$ of problems that are interreducible with \textsc{\#Bis}. Problems in the second class do not admit an fpras unless \NP=\RP,\ while the approximability status of problems in the third class is unknown and the conjecture is that they are neither interreducible with \sSAT, nor they admit an fpras.
We will denote \FPRAS\ the class of \sP\ problems that admit an fpras. 

Several works have attempted to provide a structural characterization that exactly captures \FPRAS,  in terms of path counting~\cite{BCPPZ17,PZ06}, interval size functions~\cite{BGPT17}, or  descriptive complexity~\cite{Arenas}.
Since counting problems with \NP-complete decision version are inapproximable unless $\NP=\RP$~\cite{DGGJ04}, those that admit fpras should be found among those with easy decision version (i.e., in \BPP\ or even in \cP). Even more specifically, in search of a logical characterization of a class that exactly captures \FPRAS,\ Arenas \textsl{et al.}~\cite{Arenas} show that subclasses of \FPRAS\ are contained in \totp,\ and they implicitly propose to study subclasses of \totp\ with certain additional properties in order to come up with approximable problems. Notably, most problems proven so far to admit an fpras belong to \totp,\ and several counting complexity classes proven to admit an fpras, namely $\#\mathsf{\Sigma_1}$, $\#\mathsf{R\Sigma_2}$~\cite{Saluja}, $\mathsf{\Sigma QSO(\Sigma_1)}$, $\mathsf{\Sigma QSO(\Sigma_1[FO])}$~\cite{Arenas} and $\mathsf{spanL}$~\cite{Arenas19}, are subclasses of \totp.

 Counting problems in \sP\ have also been studied in terms of descriptive complexity~\cite{Arenas,Bulatov,Dalmau,DGGJ04,Saluja}. Arenas \textsl{et al.}~\cite{Arenas} raised the question of defining  classes in terms of descriptive complexity that capture either \totp \  or robust subclasses of \totp, \ as one of the most important open questions in the area. A robust class of counting problems needs either to have a natural complete problem or to be closed under addition, multiplication and subtraction by one~\cite{Arenas}. In particular, \totp \ satisfies both of the above properties \cite{Arenas,BCPPZ17}. 

\subsubsection{Our contribution.} In the first part of the paper we focus on the exploration of the structure of \sP\ through descriptive complexity. In particular, we define two subclasses of \totp, namely $\mathsf{\Sigma QSO(\Sigma_2\text{-}2SAT)}$ and $\#\mathsf{\Pi_2\text{-}1VAR}$, via logical characterizations; for both these classes we show robustness by providing natural complete problems for them. Namely, we prove that the problem \textsc{\#Disj2Sat} of computing the number of satisfying assignments to disjunctions of 2SAT formulae is complete for $\mathsf{\Sigma QSO(\Sigma_2\text{-}2SAT)}$ under parsimonious reductions. This reveals that problems hard for $\mathsf{\Sigma QSO(\Sigma_2\text{-}2SAT)}$ under parsimonious reductions cannot admit an fpras unless $\NP=\RP$. We also prove that \textsc{\#MonotoneSat} is complete for $\#\mathsf{\Pi_2\text{-}1VAR}$ under product reductions. Our result is the first completeness result for \textsc{\#MonotoneSat} under reductions stronger than Turing. Notably, the complexity of \textsc{\#MonotoneSat} has been investigated in~\cite{HHKW07,BGPT17} and it is still open whether it is complete for \totp, or for a subclass of \totp\ under reductions for which the class is downwards closed. Although, $\#\mathsf{\Pi_2\text{-}1VAR}$ is not known to be downwards closed under product reductions, our result is a step towards understanding the exact complexity of \textsc{\#MonotoneSat}. 

\begin{wrapfigure}{r} {0.4\textwidth}
	\vspace{-20pt}
	\centering
\begin{tikzpicture}
\node at (1,3) {\sP};
\draw [thick, -> ] (1,2.3) -- (1,2.7);
\node at (1,2) {\sBPP};
\draw [thick, ->] (0.5,1.3) -- (0.9,1.7);
\node at (0,1) {\sPE};
\draw [thick, ->] (2,0.4) -- (1.3,1.7);
\node at (2,0) {\FPRAS};
\node at (0,0) {\totp};
\draw [thick, ->] (0,0.3) -- (0,0.7);
\draw [thick, <-] (0.5,0) -- (1.3,0) node[midway,above] {?};
\end{tikzpicture} 
\caption{Relation of \FPRAS\ to counting classes below \sP.}
	\vspace{-20pt}
	\label{fig:conjecture}
\end{wrapfigure}
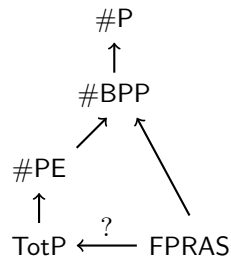

In the second part  of this paper we examine the relationship between the class \totp\ and \FPRAS.\ 
As we already mentioned, most (if not all) problems proven so far to admit fpras belong to \totp, so we would like to examine whether 
$\FPRAS\ \subseteq \totp$. Of course, problems in \FPRAS\ have decision version in \BPP~\cite{gill}, so if we assume $\cP\neq \BPP$ this is probably not the case. Therefore, a  more realistic goal is to determine assumptions under which the conjecture $\FPRAS \subseteq \totp$ might be true.
The world so far is depicted in Figure~\ref{fig:conjecture}, where \sBPP\ denotes the class of problems in \sP\ with decision version in \BPP.

 In this work we refine this picture by proving that (a) $\FPRAS \nsubseteq \totp$ 
 unless \RP=\cP, which means that proving $\FPRAS \subseteq \totp$ would be at least as hard as proving $\RP = \cP$, (b) $\totp\ \nsubseteq \FPRAS$ if and only if $\NP\neq \RP$, (c)
 \FPRAS\ lies between two classes that can be seen as counting versions of \RP\ and \BPP,
 and (d) \FPRASP, which is the subclass of \FPRAS\ with zero error probability when the function value is zero, lies between two classes that we introduce here, that can both be seen as counting versions of \RP, and which surprisingly do not coincide unless \RP=\NP. Finally,  we give a complete picture of inclusions among all the classes defined or discussed in this paper with respect to  different conjectures about the \NP\ vs. \RP\ vs. \cP\ questions.

\section{Two robust subclasses of \totp}\label{descriptive section}

In this section we give the logical characterization of two robust subclasses of \totp. Each one of them has a natural complete problem. Two kinds of reductions will be used for the completeness results; parsimonious and product reductions. Note that both of them preserve approximations of multiplicative error \cite{DGGJ04, Saluja}. 

\begin{definition}
Let $f$, $g:\Sigma ^*\rightarrow \mathbb{N}$ be two counting functions.

(a) We say that there is a parsimonious (or Karp) reduction from $f$ to $g$, symb.\ $f \le_m^p g$,  if there is a polynomial-time computable function $h$, such that for every $x\in\Sigma^*$ it holds that $f(x) = g(h(x))$.

(b) We say that there is a product reduction from $f$ to $g$, symb.\ $f\le_{pr} g$, if there are polynomial-time computable functions $h_1,h_2$ such that for every $x\in\Sigma^*$ it holds that $f(x) = g(h_1(x))\cdot h_2(x)$.
\end{definition}

The formal definitions of the classes \sP, \FP, \sPE\ and \totp\ follow.

\begin{definition}
(a) \cite{Valiant79}
\sP\ is the class of functions $f$ for which there exists a polynomial-time decidable binary relation $R$ and a polynomial $p$  such that for all $x\in\Sigma^*$,  $f(x)=\big|\{y \in \{0,1\}^* \mid  |y|=p(|x|) \text{ and }  R(x,y)\}\big|$. 

Equivalently, $\sP=\{acc_M:\Sigma^*\rightarrow \mathbb{N} \ | \  M\text{ is an NPTM}\}$.

(b) \FP\ is the class of functions in \sP\ that are computable in polynomial time. 

(c) \cite{PZ06}
	$\sPE=\{f: \Sigma^*\rightarrow \mathbb{N} \ | \ f\in\sP \text{ and } L_f \in \cP \}$, where $L_f=\{x\in \Sigma^* \mid f(x)>0\}$ is the decision version of the function $f$.
	
(d) \cite{PZ06}
$\totp =\{tot_M:\Sigma^*\rightarrow \mathbb{N} \ | \  M\text{ is an NPTM}\}$, where 
$tot_M(x)= \#($all computation paths of $M$ on input $x)-1$.
\end{definition} 

\subsection{The class $\mathsf{\Sigma QSO(\Sigma_2\text{-}2SAT)}$}\label{sub2.1}

In order to define the first class we make use of the framework of Quantitative Second-Order Logics (QSO) defined in~\cite{Arenas}. 

Given a relational vocabulary $\sigma$, 
the set of First-Order logic formulae over $\sigma$ is given by the grammar:
$$\phi:= x=y \, | \, R(\overrightarrow{x}) \, | \, \neg\phi \, | \, \phi\vee\phi \, | \, \exists x\phi \, | \, \top \, | \perp$$
where $x,y$ are first-order variables, $R\in\sigma$, $\overrightarrow{x}$ is a tuple of first order variables, $\top$ represents a tautology, and $\perp$ represents the negation of a tautology.

We define a literal to be either of the form $X(\overrightarrow{x})$ or $\neg X(\overrightarrow{x})$, where $X$ is a second-order variable and $\overrightarrow{x}$ is a tuple of first-order variables. A 2SAT clause over $\sigma$ is a formula of the form $\phi_1\vee\phi_2\vee\phi_3$, where each of the $\phi_i$'s, $1\leq i\leq 3$, can be either a literal or a first-order formula over $\sigma$. In addition, at least one of them is a first-order formula. 
The set of $\Sigma_2$-2SAT formulae over $\sigma$ are given by:
$$\psi:=\exists \overrightarrow{x}\forall \overrightarrow{y} \bigwedge_{j=1}^k C_j(\overrightarrow{x},\overrightarrow{y})$$
where $\overrightarrow{x},\overrightarrow{y}$ are tuples of first-order variables, $k\in\mathbb{N}$ and $C_j$ are 2SAT clauses for every $1\leq j\leq k$.

The set of $\mathit{\Sigma QSO(\Sigma_2\text{-}2SAT)}$ formulae over $\sigma$ is given by the following grammar:
$$\alpha:=\phi\, | \, s \, |  \, (\alpha+\alpha) \, | \, \Sigma x. \alpha \, | \, \Sigma X.\alpha$$
where $\phi$ is a $\Sigma_2$-2SAT formula, $s\in\mathbb{N}$, $x$ is a first-order variable and $X$ is a second-order variable. The syntax of $\mathit{\Sigma QSO(\Sigma_2\text{-}2SAT)}$ formulae includes the counting operators of addition $+$, $\Sigma x$, $\Sigma X$. Specifically, $\Sigma x$, $\Sigma X$ are called first-order and second-order quantitative quantifiers respectively.

\begin{table}
\centering
$[[\phi]](\mathcal{A},v,V)=\begin{cases}
1, \text{ if } \mathcal{A}\models\phi\\
0, \text{ otherwise}
\end{cases}$
\hspace{7mm}
$[[s]](\mathcal{A},v,V)=s$

\vspace{2mm}

$[[\alpha_1+\alpha_2]](\mathcal{A},v,V)=[[\alpha_1]](\mathcal{A},v,V)+[[\alpha_2]](\mathcal{A},v,V)$

\vspace{2mm}

$\displaystyle[[\Sigma x. \alpha]](\mathcal{A},v,V)=\sum_{a\in A}\,[[\alpha]](\mathcal{A},v[a/x],V)$

\vspace{2mm}

$\displaystyle [[\Sigma X. \alpha]](\mathcal{A},v,V)=\sum_{B\subseteq A^{arity(X)}}[[\alpha]](\mathcal{A},v,V[B/X])$
\vspace{2mm}
\caption{The semantics of $\mathit{\Sigma QSO(\Sigma_2\text{-}2SAT)}$ formulae}
\label{semantics}
\end{table}

Let $\sigma$ be a relational vocabulary, $\mathcal{A}$ a $\sigma$-structure with universe $A$, $v$ a first-order assignment for $\mathcal{A}$ and $V$ a second-order assignment for $\mathcal{A}$. Then the evaluation of a $\mathit{\Sigma QSO(\Sigma_2\text{-}2SAT)}$ formula $\alpha$ over $(\mathcal{A}, V, v)$ is defined as a function $[[\alpha]]$ that on input $(\mathcal{A}, V, v)$ returns a number in $\mathbb{N}$. The function $[[\alpha]]$ is recursively defined in Table \ref{semantics}. A $\mathit{\Sigma QSO(\Sigma_2\text{-}2SAT)}$ formula
$\alpha$ is said to be a sentence if it does not have any free variable, that is, every variable in $\alpha$ is under the scope of a usual quantifier ($\exists$, $\forall$) or a quantitative quantifier. It is important
to notice that if $\alpha$ is a $\mathit{\Sigma QSO(\Sigma_2\text{-}2SAT)}$ sentence over a vocabulary $\sigma$, then for every $\sigma$-structure $\mathcal{A}$,
first-order assignments $v_1$, $v_2$ for $\mathcal{A}$ and second-order assignments $V_1$, $V_2$ for $\mathcal{A}$, it holds
that $[[\alpha]](\mathcal{A}, v_1, V_1) = [[\alpha]](\mathcal{A}, v_2, V_2)$. Thus, in such a case we use the term $[[\alpha]](\mathcal{A})$ to denote $[[\alpha]](\mathcal{A}, v, V)$ for some arbitrary first-order assignment $v$ and some arbitrary second-order assignment $V$ for $\mathcal{A}$.

At this point it is clear that for any $\mathit{\Sigma QSO(\Sigma_2\text{-}2SAT)}$ formula  $\alpha$, a function $[[\alpha]]$ is defined. In the rest of the paper we will use the same notation, namely $\mathsf{\Sigma QSO(\Sigma_2\text{-}2SAT)}$, both for  the set of formulae and the set of corresponding counting functions.\footnote{Moreover, we will
use the terms `(counting) problem' and `(counting) function' interchangeably throughout the paper.
}

The following inclusion holds between the class $\mathsf{\#RH\Pi_1}$~\cite{DGGJ04} and the class $\mathsf{\Sigma QSO(\Sigma_2\text{-}2SAT)}$ defined presently.

\begin{proposition}\label{prop1}
$\mathsf{\#RH\Pi_1}\subseteq \mathsf{\Sigma QSO(\Sigma_2\text{-}2SAT)}$
 \end{proposition}
 \begin{proof}
A function $f$ is in the class $\#\mathsf{RH\Pi_1}$ if it can be expressed in the form 
$f(\mathcal{A})=|\{\langle \overrightarrow{X},\overrightarrow{x} \rangle: \mathcal{A}\models \forall \overrightarrow{y}\psi(\overrightarrow{y},\overrightarrow{x},\overrightarrow{X})\}|$, where $\psi$ is an unquantified CNF formula in which each clause has at most one occurrence of an unnegated variable from $\overrightarrow{X}$, and at most one occurrence of a negated variable from $\overrightarrow{X}$. Alternatively, the function $f$ can be expressed in the form $[[\Sigma\overrightarrow{X}.\Sigma\overrightarrow{x}.\forall\overrightarrow{y}\psi(\overrightarrow{y},\overrightarrow{x},\overrightarrow{X})]](\mathcal{A}).$
The Restricted-Horn formula $\psi$ is also a 2SAT formula. 

Therefore, $f\in\mathsf{\Sigma QSO(\Sigma_2\text{-}2SAT)}$.
\end{proof}

\medskip
The class $\mathsf{\Sigma QSO(\Sigma_2\text{-}2SAT)}$ contains problems that are tractable, such as \textsc{\#2Col}, which is known to be computable in polynomial time~\cite{GGJ17}. It also contains all the problems in $\mathsf{\#RH\Pi_1}$, such as \textsc{\#Bis}, \textsc{\#1P1NSat}, \textsc{\#Downsets}~\cite{DGGJ04}. These three problems are complete for $\mathsf{\#RH\Pi_1}$ under approximation preserving reductions and are not believed to have an fpras. At last, the problem $\textsc{\#Is}$~\cite{DGGJ04}, which is interriducible with \sSAT\ under approximation preserving reductions, belongs to $\mathsf{\Sigma QSO(\Sigma_2\text{-}2SAT)}$ as well.

We next show that a generalization of \textsc{\#2Sat}, which we will call \textsc{\#Disj2Sat}, is complete for $\mathsf{\Sigma QSO(\Sigma_2\text{-}2SAT)}$ under parsimonious reductions.
 
\subsubsection{Membership of \textsc{\#Disj2Sat} in $\mathsf{\Sigma QSO(\Sigma_2\text{-}2SAT)}$ \medskip \\}

In propositional logic, a 2SAT formula is a conjunction of clauses that contain at most two literals. Suppose we are given a propositional formula $\phi$, which is a disjunction of 2SAT formulae, then \textsc{\#Disj2Sat} on input $\phi$ equals the number of satisfying assignments of $\phi$. 

In this subsection we assume that 2SAT formulae consist of clauses which contain exactly two literals since we can rewrite a clause of the form $l$ as $l\vee l$, for any literal $l$.

\begin{theorem}\label{membership1}
$\textsc{\#Disj2Sat} \in \mathsf{\Sigma QSO(\Sigma_2\text{-}2SAT)}$
\end{theorem}
\begin{proof}
Consider the vocabulary $\sigma=\{C_1,C_2,C_3,C_4,D\}$ where $C_i$, $1\leq i\leq 4$, are ternary relations and $D$ is a binary relation. This vocabulary can encode any formula which is a disjunction of 2SAT formulae. More precisely, $C_1(c,x,y)$ iff clause $c$ is of the form $x\vee y$, $C_2(c,x,y)$ iff $c$ is $\neg x\vee y$, $C_3(c,x,y)$ iff $c$ is $x\vee\neg y$, $C_4(c,x,y)$ iff $c$ is $\neg x\vee\neg y$ and $D(d,c)$ iff clause $c$ appears in the ``disjunct'' $d$.

Let $\phi$ be an input to \textsc{\#Disj2Sat} encoded by an ordered $\sigma$-structure $\mathcal{A}=\langle A, C_1, C_2, C_3,C_4, D\rangle$, where the universe $A$ consists of elements representing variables, clauses and ``disjuncts''. Then, it holds that the number of satisfying assignments of $\phi$ is equal to $[[\Sigma T.\psi(T)]](\mathcal{A})$, where\\

\begin{math}
\begin{array}{ll}
\psi(T):=\exists d\forall c\forall x\forall y & \big((\neg D(d,c)\vee\neg C_1(c,x,y)\vee T(x)\vee T(y))\wedge\\
&(\neg D(d,c)\vee\neg C_2(c,x,y)\vee \neg T(x)\vee T(y)) \wedge\\
&(\neg D(d,c)\vee\neg C_3(c,x,y)\vee T(x)\vee \neg T(y)) \wedge\\
&(\neg D(d,c)\vee\neg C_4(c,x,y)\vee \neg T(x)\vee \neg T(y)\big)
\end{array}\\
\end{math}

Thus, $\#\textsc{Disj2Sat}$ is defined by $\Sigma T.\psi(T)$ which is in $\mathsf{\Sigma QSO(\Sigma_2\text{-}2SAT)}$.
\end{proof}

\subsubsection{Hardness of \textsc{\#Disj2Sat}\medskip\\} 

Suppose we have a formula $\alpha$ in $\mathsf{\Sigma QSO(\Sigma_2\text{-}2SAT)}$ and an input structure $\mathcal{A}$ over a vocabulary $\sigma$. We describe a polynomial-time reduction that given $\alpha$ and $\mathcal{A}$, it returns a propositional formula $\phi_{\alpha_\mathcal{A}}$ which is a disjunction of 2SAT formulae and it holds that $[[\alpha]](\mathcal{A})=\#\textsc{Disj2Sat}(\phi_{\alpha_\mathcal{A}})$. The reduction is a parsimonious reduction, i.e. it preserves the values of the functions involved.

\begin{theorem}\label{hardness of disj2sat}
\textsc{\#Disj2Sat} is hard for $\mathsf{\Sigma QSO(\Sigma_2\text{-}2SAT)}$ under parsimonious reductions.
\end{theorem}
\begin{proof}
By Proposition 5.1 of~\cite{Arenas}, $\alpha$ can be written in the form \\
$\displaystyle \sum_{i=1}^{m}\Sigma \overrightarrow{X}_i.\Sigma \overrightarrow{x}.\exists \overrightarrow{y}\forall \overrightarrow{z} \bigwedge_{j=1}^{n} C_j^i(\overrightarrow{X}_i,\overrightarrow{x},\overrightarrow{y},\overrightarrow{z})$,  where each $\overrightarrow{X}_i$ is a sequence of second-order variables and each $C_j^i$ is a 2SAT clause. Each term of the sum can be replaced by $\displaystyle \Sigma \overrightarrow{X}.\Sigma \overrightarrow{x}.\exists \overrightarrow{y}\forall \overrightarrow{z} \bigwedge_{j=1}^{n} C_j^i(\overrightarrow{X}_i,\overrightarrow{x},\overrightarrow{y},\overrightarrow{z})\wedge \bigwedge_{X\not\in\overrightarrow{X}_i} \forall \overrightarrow{u} X(\overrightarrow{u})$ where $\overrightarrow{X}$ is the union of all $\overrightarrow{X}_i$. Now we have expressed $\alpha$ in the following form \\
$\displaystyle \sum_{i=1}^{m}\Sigma \overrightarrow{X}.\Sigma \overrightarrow{x}.\exists \overrightarrow{y}\forall \overrightarrow{z} \bigwedge_{j=1}^{n} \phi_j^i(\overrightarrow{X},\overrightarrow{x},\overrightarrow{y},\overrightarrow{z})$. \\
The next step is to expand the first-order quantifiers and sum operators and replace their variables with first-order constants from the universe $A$. 

In this way, we obtain $\alpha_{\mathcal{A}}:=
\displaystyle \sum_{i=1}^{m}\sum_{\overrightarrow{a}\in A^{|\overrightarrow{x}|}}\Sigma \overrightarrow{X}.\bigvee_{\overrightarrow{b}\in A^{|\overrightarrow{y}|}} \bigwedge_{i=1}^n\bigwedge_{\overrightarrow{c}\in A^{|\overrightarrow{z}|}} \phi_j^i(\overrightarrow{X},\overrightarrow{a},\overrightarrow{b},\overrightarrow{c})$. Each first-order subformula of $\phi_j^i$ has no free-variables and is either satisfied or not satisfied by $\mathcal{A}$, so we can replace it by $\top$ or $\perp$ respectively. Also, after grouping the sums and the conjunctions, we get $\displaystyle \sum_{i=1}^{m'}\Sigma \overrightarrow{X}.\bigvee_{j=1}^{n_1}\bigwedge_{k=1}^{n_2}\psi_{j,k}^i(\overrightarrow{X})$. The formulae $\psi_{j,k}^i(\overrightarrow{X})$ are conjunctions of clauses that consist of $\perp$, $\top$ and at most two literals of the form $X_t(\overrightarrow{a}_l)$ or $\neg X_t(\overrightarrow{a}_l)$ for some second-order variable $X_t$ and some tuple of first-order constants $\overrightarrow{a}_l$. We can eliminate the clauses that contain a $\top$ and remove $\perp$ from the clauses that contain it. After this simplification, some combinations of variable-constants may not appear in the remaining formula. For any such combination $X(\overrightarrow{a})$, we add a clause $\psi_{X,\overrightarrow{a}}:=X(\overrightarrow{a})\vee\neg X(\overrightarrow{a})$, since $X(\overrightarrow{a})$ can have any truth value. 

So, we have reformulated the above formula and we get $\displaystyle \sum_{i=1}^{m'}\Sigma \overrightarrow{X}.\bigvee_{j=1}^{n_1}\bigwedge_{k=1}^{n_2'}\psi_{j,k}^i(\overrightarrow{X})$. After replacing every appearance of $X_t(\overrightarrow{a}_l)$ by a propositional variable $x_{tl}$, the part $\displaystyle\bigvee_{j=1}^{n_1}\bigwedge_{k=1}^{n_2'}\psi_{j,k}^i(\overrightarrow{X})$ becomes a disjunction of 2SAT formulae. Finally, we introduce $m'$ new propositional variables $x_1,...x_{m'}$ and define\\
 $\displaystyle\phi_{\alpha_\mathcal{A}}:=\bigvee_{i=1}^{m'}\bigvee_{j=1}^{n_1}\bigwedge_{k=1}^{n_2'}\psi_{j,k}^i\wedge x_i\bigwedge_{l\neq i} \neg x_l$. The formula $\phi_{\alpha_\mathcal{A}}$ is a disjunction of 2SAT formulae and the number of its satisfying assignments is equal to $[[\alpha]](\mathcal{A})$. Moreover, every transformation we made requires polynomial time in the size of the input structure $\mathcal{A}$. 
\end{proof}
 
 \medskip
 It is known that \textsc{\#2Sat} has no fpras unless $\NP=\RP$, since it is equivalent to counting all independent sets in a graph~\cite{DGGJ04}. Thus,  problems hard for $\mathsf{\Sigma QSO(\Sigma_2\text{-}2SAT)}$ under parsimonious reductions also cannot admit an fpras unless $\NP=\RP$.

 \subsubsection{Inclusion of $\mathsf{\Sigma QSO(\Sigma_2\text{-}2SAT)}$ in $\totp$\medskip\\} 
 
Several problems in $\mathsf{\Sigma QSO(\Sigma_2\text{-}2SAT)}$, like \textsc{\#1P1NSat}, \textsc{\#Is}, \textsc{\#2Col}, and \textsc{\#2Sat}, are also in \totp. We next prove that this is not a coincidence. 
 
 \begin{theorem}
 $\mathsf{\Sigma QSO(\Sigma_2\text{-}2SAT)}\subseteq \totp$
 \end{theorem}
 \begin{proof}
 Since $\totp$ is exactly the Karp closure of self-reducible functions of $\sPE$~\cite{PZ06}, it suffices to show that the $\mathsf{\Sigma QSO(\Sigma_2\text{-}2SAT)}$-complete problem  \#\textsc{Disj2Sat} is such a function. 
  
  First of all, \textsc{Disj2Sat} belongs to \cP. Thus \#\textsc{Disj2Sat} $\in\sPE$.
 
 Secondly, every counting function associated with the problem of counting satisfying assignments for a propositional formula is self-reducible~\footnote{\totp\ contains all self-reducible problems in \sP,\ with decision version in \cP. 
Intuitively, self-reducibility means that counting the number of solutions to an instance of a problem, can be performed recursively by computing the number of solutions to some other instances of the same problem. For example, \sSAT\ is self-reducible: the number of satisfying assignments of a formula $\phi$ is equal to  the sum of the number of satisfying assignments of $\phi_1$ and $\phi_0$, where $\phi_i$ is $\phi$ with its first variable fixed to $i$.}.
 So \#\textsc{Disj2Sat} has this property as well.
 
  Therefore, any  $\mathsf{\Sigma QSO(\Sigma_2\text{-}2SAT)}$ formula $\alpha$ defines a function $[[\alpha]]$ that belongs to \totp.
 \end{proof}

\medskip
\begin{corollary}
 $\mathsf{\#RH\Pi_1}\subseteq\totp$
\end{corollary}

\subsection{The class $\#\mathsf{\Pi_2\text{-}1VAR}$}

To define the second class $\#\mathsf{\Pi_2\text{-}1VAR}$, we make use of the framework presented in~\cite{Saluja}.

We say that a counting problem $\#B$ belongs to the class $\#\mathsf{\Pi_2\text{-}1VAR}$ if for any ordered structure $\mathcal{A}$ over a vocabulary $\sigma$, which is an input to $\#B$, it holds that $\#B(\mathcal{A})=|\{\langle X \rangle: \mathcal{A}\models \forall \overrightarrow{y}\exists\overrightarrow{z}\psi(\overrightarrow{y},\overrightarrow{z},X)\}|$. The formula $\psi(\overrightarrow{y},\overrightarrow{z},X)$ is of the form $\displaystyle\phi(\overrightarrow{y},\overrightarrow{z})\wedge X(\overrightarrow{z})$, where $\phi$ is a first-order formula over $\sigma$ and $X$ is a positive appearance of a second-order variable. We call the formula $\psi$ a variable, since it contains only one second-order variable. Moreover, we allow counting only the assignments to the second-order variable $X$ under which the structure $\mathcal{A}$ satisfies $\forall \overrightarrow{y}\exists\overrightarrow{z}\psi(\overrightarrow{y},\overrightarrow{z},X)$.

\begin{proposition}\label{prop2}
$\textsc{\#Vc}\in\#\mathsf{\Pi_2\text{-}1VAR}$, where \textsc{\#Vc} is the problem of counting the vertex covers of all sizes in a graph.
\end{proposition}
\begin{proof}
An input graph $G$ to \textsc{\#Vc} can be encoded as a finite structure $\mathcal{G}$ using the vocabulary $\sigma=\{E, End\}$, where $E$ is the edge relation and $End$ is a binary relation. The universe is the set of all vertices and all edges. $End(u,e)$ iff vertex $u$ is an endpoint of edge $e$. Then, $\#\textsc{Vc}(\mathcal{G})=|\{\langle VC\rangle \mid \mathcal{G}\models\forall x\exists y \Big(End(y,x)\wedge VC(y)\}|$. Therefore, $\#\textsc{Vc}\in\#\mathsf{\Pi_2\text{-}1VAR}$.
\end{proof}

\subsubsection{Completeness of \#\textsc{MonotoneSat} for $\#\mathsf{\Pi_2\text{-}1VAR}$\medskip\\} \label{sub3.1}

Given  a propositional formula $\phi$ in conjunctive normal form, where all the literals are positive,  \#\textsc{MonotoneSat} on input $\phi$ equals the number of satisfying assignments of $\phi$.

\begin{theorem}\label{membership2}
$\#\textsc{MonotoneSat} \in \#\mathsf{\Pi_2\text{-}1VAR}$
\end{theorem}
\begin{proof}
Consider the vocabulary $\sigma=\{C\}$ with the binary relation $C(c,x)$ to indicate that the variable $x$ appears in the clause $c$. Given a $\sigma$-structure $\mathcal{A}=\langle A, C\rangle$ that encodes a formula $\phi$, which is an input to \#\textsc{MonotoneSat}, it holds that 
\#\textsc{MonotoneSat}$(\phi)$=$|\{\langle T\rangle: \mathcal{A}\models \forall c\exists x \big(C(c,x)\wedge T(x)\big)\}|$.

Therefore, \#\textsc{MonotoneSat} $\in$ $\#\mathsf{\Pi_2\text{-}1VAR}$.
\end{proof}

\medskip
\begin{theorem}\label{hardness of monotonesat}
\#\textsc{MonotoneSat} is hard for $\#\mathsf{\Pi_2\text{-}1VAR}$ under product reductions.
\end{theorem}
\begin{proof}
 We show that there is a polynomial-time product reduction from any  $\# B\in\#\mathsf{\Pi_2\text{-}1VAR}$ to \#\textsc{MonotoneSat}. This means that there are polynomial-time computable functions $g$ and $h$, such that for every $\sigma$-strucrure $\mathcal{A}$ that is an input to $\#B$ we have 
$\#B(\mathcal{A})=\#\textsc{MonotoneSat}\big(g(\mathcal{A})\big)\cdot h(|A|)$. 

Suppose we have a problem $\# B\in \#\mathsf{\Pi_2\text{-}1VAR}$ and a $\sigma$-structure $\mathcal{A}$. Then, there exists a formula $\psi$ of the form $\displaystyle\psi(\overrightarrow{y},\overrightarrow{z},X)=\phi(\overrightarrow{y},\overrightarrow{z})\wedge X(\overrightarrow{z})$ such that $\#B(\mathcal{A})=|\{\langle X\rangle:\mathcal{A}\models\forall\overrightarrow{y}\exists\overrightarrow{z}\psi(\overrightarrow{y},\overrightarrow{z},X)\}|$.

The formula $\forall\overrightarrow{y}\exists\overrightarrow{z}\psi(\overrightarrow{y},\overrightarrow{z},X)$ can be written in the form 

\[\displaystyle\bigwedge_{\overrightarrow{a}\in A^{|\overrightarrow{y}|}}\bigvee_{\overrightarrow{b}\in A^{|\overrightarrow{z}|}}\phi(\overrightarrow{a},\overrightarrow{b})\wedge X(\overrightarrow{b}).\]

By substituting first-order subformulae by $\top$ or $\perp$ and simplifying, we obtain
$\displaystyle \chi_{\psi_\mathcal{A}}:=\bigwedge_{i=1}^{n_1}\bigvee_{j=1}^{n_2} X(\overrightarrow{b}_{i,j})$, where each $\overrightarrow{b}_{i,j}$ is a tuple of first-order constants. To define $\chi_{\psi_\mathcal{A}}$, we have simplified the subformulae containing $\perp$ and $\top$. As a result, there may be some combinations of the second-order variable $X$ and first-order constants that do not appear in $\chi_{\psi_\mathcal{A}}$. Let $n(\mathcal{A})$ be the number of these combinations. The last transformation consists of replacing every $X(\overrightarrow{b}_{i,j})$ with a propositional variable $x_{ij}$, so we get the output of the function $g$, which is $\displaystyle g(\mathcal{A}):=\bigwedge_{i=1}^{n_1}\bigvee_{j=1}^{n_2} x_{i,j}$. This formula has no negated variables, so it can be an input to \#\textsc{MonotoneSat}. Finally, since the missing $n(\mathcal{A})$ variables can have any truth value, we have $\#B(\mathcal{A})=\#\textsc{MonotoneSat}\big(g(\mathcal{A})\big)\cdot 2^{n(\mathcal{A})}$. 
\end{proof}

\subsubsection{Inclusion of $\#\mathsf{\Pi_2\text{-}1VAR}$ in $\totp$\\}
\begin{theorem}
$\#\mathsf{\Pi_2\text{-}1VAR}\in\totp$
\end{theorem}
\begin{proof}
It is easy to prove that $\#\textsc{MonotoneSat} \in \totp$ and that \totp\ is closed under product reductions.  Thus, the above results imply that every counting problem in $\#\mathsf{\Pi_2\text{-}1VAR}$ belongs to \totp.
\end{proof}

\section{On \totp\ vs. \FPRAS}\label{fpras section}

In this section we study the relationship between the classes \totp\ and 
\FPRAS. First of all we give some definitions and facts that will be needed.

\begin{theorem}\cite{PZ06}
\label{PZmt}
(a) \FP\ $\subseteq$ \totp\ $\subseteq \sPE\ \subseteq \sP$. The inclusions are proper unless $\cP=\NP$.

(b) \totp\ is the Karp closure of self-reducible $\sPE$ functions.
\end{theorem}

We consider \FPRAS\ to be the class of functions in \sP\ that admit fpras, and we also introduce an ancillary class \FPRASP. 
Formally:

\begin{definition}   A function $f$ belongs to \FPRAS\ if $f\in \sP$ and there exists a randomized algorithm that on input $x \in \Sigma^*,$ $\epsilon>0,$ $\delta>0,$  returns a value $\widehat{f(x)}$ such that
\[\Pr[(1-\epsilon) f(x) \leq \widehat{f(x)} \leq (1+\epsilon) f(x)]\geq 1-\delta\] in time poly($|x|,\epsilon^{-1}, \log \delta^{-1}$).

We further say that a function $f\in \FPRAS$ belongs to \FPRASP\ if whenever $f(x)=0$ the returned value $\widehat{f(x)}$ equals $0$ with probability~1.
\end{definition}

We begin with the following observation.\footnote{The following theorem is probably well-known among counting complexity researchers. However, since we have not been able to find a proof in the literature we provide one here for the sake of completeness.}

\begin{theorem}\label{sP vs fpras}
$\sP\subseteq \FPRAS$ if and only if \NP=\RP.
\end{theorem}
\begin{proof} 
For the one direction we observe that if \NP$\neq$\RP\ then there are functions in \sP,\ that are not in \FPRAS.\ For example, \textsc{\#Is} belongs to \sP,\ and does not admit an fpras unless \NP=\RP\ \cite{DFJ02}.

The other direction derives from a Stockmeyer's well known theorem \cite{Stockmeyer85a}. By Stockmeyer's theorem there exists an fpras, with access to a $\Sigma_2^p$ oracle, for any problem in \sP.\ If \NP=\RP\ then $\Sigma_2^p= \RP^{\RP}\subseteq \BPP$~\cite{zachos88}. Finally it is easy to see that an fpras with access to a \BPP\ oracle, can be replaced by another fpras, that simulates the oracle calls itself.
\end{proof}

\medskip
\begin{corollary}\label{Corollary1}
$\totp\subseteq\FPRAS$ if and only if $\totp\subseteq\FPRASP$ if and only if $\NP=\RP$.
\end{corollary}
\begin{proof}
$\totp\subseteq\FPRAS$ iff \NP=\RP\ is an immediate corollary of the proof of Theorem~\ref{sP vs fpras} along with the observations that $\textsc{\#Is}\in\totp$ and $\totp\subseteq\sP$.

We prove that $\totp\subseteq\FPRAS$ iff $\totp\subseteq\FPRASP$. Suppose that $\totp\subseteq\FPRAS$ and let $f$ be a function in \totp. Then $f\in\FPRAS$. Now we can modify the fpras for $f$ so that it returns the correct value of $f(x)$ with probability~1 if $f(x)=0$. We can do this since we can decide if $f(x)=0$ in polynomial time. So, $f\in\FPRASP$.

The other direction is trivial by the inclusion $\FPRASP\subseteq\FPRAS$.
\end{proof}

\medskip
Now we examine the opposite inclusion, i.e.\ whether \FPRAS\ is a subset of \totp.\
To this end we introduce two classes that contain counting problems with decision in \RP.

 Recall that if a counting function $f$ admits an fpras, then its decision version, i.e.\ deciding whether $f(x)=0$, is in \BPP. In a similar way, if a counting function belongs to \FPRASP, then its decision version is in \RP. So we need to define the subclass of \sP\ with decision in \RP. Clearly, if for a problem $\Pi$ in \sP\ the corresponding counting machine has an \RP\ behavior (i.e., either a majority of paths are accepting or all paths are rejecting) then the decision version is naturally in \RP. However, this seems to be a quite restrictive requirement. Therefore we will examine two subclasses of \sP. 
 
 For that we need the following definition of the set of Turing Machines associated to problems in \RP. 
 
 \begin{definition}\label{MR}
Let $M$ be an NPTM. We denote by $p_M$ the polynomial such that  on inputs of size $n$, $M$ makes $p_M(n)$ non-deterministic choices.\\
${\cal MR} = \{ M \mid M\text{ is an NPTM and  for all } x\in\Sigma^*$ either $acc_M(x)=0$ or $acc_M > \frac{1}{2}\cdot 2^{p_M(|x|)}\}.$  
\end{definition}

\begin{definition}
\RPo\ = $\{f\in \sP \mid \exists M \in {\cal MR}\, \forall x\in\Sigma^*: f(x) = acc_M(x)\}.$
\end{definition}

\begin{definition}
\RPt\ = $\{f \in \sP \mid L_f \in \RP \}.$
\end{definition}

Note that \RPo, although restrictive, contains counting versions of some of the most representative problems in \RP, for which no deterministic algorithms are known. For example consider the polynomial identity testing  problem (\textsc{Pit}~\footnote{Determining the computational complexity of polynomial identity testing is considered one of the most important open problems in the mathematical field of Algebraic Computing Complexity.}):
Given an arithmetic circuit of degree $d$ that computes a polynomial in a field, determine whether the polynomial is not equal to the zero polynomial. A probabilistic solution to it is to evaluate it on a random point (from a sufficiently large subset $S$ of the field). If the polynomial is zero then all points will be evaluated to $0,$ else the probability of getting $0$ is at most $\frac{d}{|S|}$. A counting analogue of \textsc{Pit} is to count the number of elements in $S$ that evaluate to non-zero values; clearly this problem belongs to \RPo. Another problem in \RPo\ is to count the number of compositeness witnesses (as defined by the Miller-Rabin primality test) on input an integer $n>2$; although in this case the decision problem is in \cP\ (a prime number has no such witnesses and this can be checked deterministically by AKS algorithm~\cite{AKS}), for a composite number $n$ at least half of the integers in $\mathbb{Z}_n$ are Miller-Rabin witnesses, hence there exists a NPTM $M \in {\cal MR}$ that has as many accepting paths as the number of witnesses.

\RPt\ contains natural counting problems as well. Two examples in \RPt\ are  \#\textsc{Exact Matchings} and \#\textsc{Blue-Red Matchings}, which are counting versions of \textsc{Exact Matching}~\cite{PY82} and \textsc{Blue-Red Matching}~\cite{NPZ07}, respectively, both of which belong to \textsf{RP} (in fact in \textsf{RNC}) as shown in~\cite{MVV87,NPZ07}; however, it is still open so far whether they can be solved in polynomial time. Therefore it is also open whether \#\textsc{Exact Matchings} and \#\textsc{Blue-Red Matchings} belong to \totp.

We will now focus on relationships  among the  aforementioned classes. We start by presenting some unconditional inclusions and then we explore possible inclusions under the condition that either $\NP\neq\RP\neq\cP$ or $\NP\neq\RP=\cP$ holds.

The results are summarized in Figures~\ref{fig1} and~\ref{fig2}.  

\begin{figure}[ht]
\centering
\begin{minipage}{0.28\textwidth}
\vfill\centering
\begin{tikzpicture}
\node at (1,4) {\sP};
\draw [thick, -> ] (1,3.3) -- (1,3.7);
\node at (1,3) {\sBPP};
\draw [thick, -> ] (0.5,2.3) -- (0.9,2.7);
\draw [thick, -> ] (1.5,2.3) -- (1.1,2.7);
\node at (0,2) {\RPt};
\draw [thick, ->] (0,1.3) -- (0,1.7);
\node at (2,2) {\FPRAS};
\node at (0,1) {\sPE};
\draw [thick, ->] (2,1.3) -- (2,1.7);
\draw [thick, ->] (1.5,1.3) -- (0.5,1.7);
\node at (2,1) {\FPRASP};
\draw [thick, ->] (0,0.3) -- (0,0.7);
\node at (0,0) {\totp};
\draw [thick, ->] (2,0.3) -- (2,0.7);
\node at (2,0) {\RPo};
\node at (1,-1) {\FP};
\draw [thick, ->] (0.8,-0.7) -- (0.5,-0.3);
\draw [thick, ->] (1.2,-0.7) -- (1.5,-0.3);
\end{tikzpicture} 
\caption{Unconditional inclusions.}
\label{fig1}
\end{minipage}
\hfill
\begin{minipage}{0.6\textwidth}
\centering
\begin{tabular}{|c|c|}
\hline 
\begin{tikzpicture}
\node at (1,4.5) {\textbf{\small{\NP$\neq$\RP$\neq$\cP}}};
\node at (0,3.5) {\sP};
\node at (2,3) {\sBPP};
\draw [thick, -> ] (1.3,3.2) -- (0.6,3.5);
\draw [thick, -> ] (2,2.3) -- (2,2.7);
\node at (0,2) {\RPt};
\draw [thick, ->] (0.7,2.3) -- (1.3,2.8);
\draw [thick, |->] (0,2.5) -- (0,3.1);
\draw [thick, |->] (0,1.3) -- (0,1.7);
\node at (0,1) {\sPE};
\draw [thick, ->] (2,1.3) -- (2,1.7);
\node at (2,2) {\FPRAS};
\draw [thick, -|] (0.7,2) -- (1.3,2);
\draw [thick, |->] (1.4,2.3) -- (0.5,3.2);
\draw [thick, |-|] (0.7,1) -- (1.3,1);
\node at (2,1) {\FPRASP};
\draw [thick, |->] (1.2,1.3) -- (0.6,1.7);
\draw [thick, |->] (0,0.3) -- (0,0.7);
\node at (0,0) {\totp};
\draw [thick, |-|] (0.7,0) -- (1.3,0);
\draw [thick, ->] (2,0.3) -- (2,0.7);
\node at (2,0) {\RPo};
\draw [thick, |-|] (0.7,0.3) -- (1.5,1.7);
\draw [thick, |-|] (1.3,0.3) -- (0.7,0.7);
\node at (1,-1) {\FP};
\draw [thick, |->] (0.8,-0.7) -- (0.5,-0.3);
\draw [thick, |->] (1.2,-0.7) -- (1.5,-0.3);
\end{tikzpicture} 
& 
\begin{tikzpicture}
\node at (1,4.5) {\textbf{\small{\NP$\neq$\RP=\cP}}};
\node at (0,3.5) {\sP};
\draw [thick, |-> ] (0,2.5) -- (0,3.1);
\node at (2,3) {\sBPP};
\draw [thick, -> ] (1.3,3.2) -- (0.6,3.5);
\draw [thick, -> ] (2,2.3) -- (2,2.7);
\node at (0,2) {\sPE=\RPt};
\draw [thick, ->] (0.7,2.3) -- (1.3,2.8);
\draw [thick, -|] (1,2) -- (1.4,2);
\draw [thick, |->] (0,1.3) -- (0,1.7);
\node at (2,2) {\FPRAS};
\draw [thick, |->] (1.4,2.3) -- (0.5,3.2);
\node at (0,1) {\totp};
\draw [thick, -|] (0.7,1.2) -- (1.4,1.7);
\draw [thick, ->] (2,1.3) -- (2,1.7);
\draw [thick, -|] (0.7,1) -- (1.3,1);
\node at (2,1) {\FPRASP};
\draw [thick, |->] (1.3,1.3) -- (0.7,1.7);
\draw [thick, ->] (2,0.3) -- (2,0.7);
\node at (2,0) {\RPo};
\draw [thick, |-] (1.3,0.3) -- (0.7,0.7);
\node at (1,-1) {\FP};
\draw [thick, ->] (1.2,-0.7) -- (1.5,-0.3);
\draw [thick, |->] (0.8,-0.7) -- (0.1,0.7);
\end{tikzpicture} \\ 
\hline 
\end{tabular}
\caption{Conditional inclusions. The following notation is used: 
$A \rightarrow B$ denotes $A \subseteq B$,  $A\dashv B$ denotes $A \not\subseteq B$, and $A \mapsto B$ denotes $A \subsetneq B$.}
\label{fig2}
\end{minipage}
\end{figure}

\subsection{Unconditional inclusions}

\begin{theorem}\label{line of inclusions}
$\FP\subseteq \RPo \subseteq \RPt \subseteq \sP$. Also $\totp\subseteq \sPE\subseteq \RPt$.
\end{theorem}
\begin{proof}
Let $f\in \FP$. We will show that $f\in \RPo$. We will construct an NPTM $M\in {\cal MR}$ s.t. on input $x$, $acc_M(x)=f(x)$. Let $x\in \Sigma^*$. We construct $M$ that computes $f(x)$ and then it computes $i\in \mathbb{N}$ s.t. $f(x)\in (2^{i-1},2^i]$. $M$ makes $i$ non-deterministic choices $b_1,b_2,...,b_i$. Each such $b\in \{0,1\}^i$ determines a path, in particular, $b$ corresponds to the $(b+1)$-st path (since $0^i$ is the first path). $M$ returns yes iff $b+1 \leq f(x)$, so $acc_M=f(x)$. Since $f(x) > 2^{i-1},$ $M\in{\cal MR}.$ 

The other inclusions are immediate by definitions.
\end{proof}

\medskip

\begin{theorem}\label{main theorem}
$\RPo \subseteq \FPRASP \subseteq \RPt.$
\end{theorem}
\begin{proof}
For the first inclusion, let $\epsilon > 0, \delta >0.$ Let $f\in \RPo$. There exists an $M_f\in {\cal MR}$ s.t. $\forall x$, $acc_{M_f}(x)=f(x).$ Let $q(|x|)$ be the number of non-deterministic choices of $M_f.$ Let $p=\frac{f(x)}{2^{q(|x|)}}$. We can compute an estimate $\hat{p}$ of $p,$ by choosing $m=poly(\epsilon^{-1}, \log \delta^{-1})$ paths uniformly at random. Then we can compute $\widehat{f(x)}=\hat{p}\cdot 2^{q(|x|)}$.

To proceed with the proof we need the following lemma.
\begin{lemma}\label{unbiased estimator}
(Unbiased estimator.) Let $A\subseteq B$ be two finite sets, and let $p=\frac{|A|}{|B|}$. Suppose we take $m$ samples from $B$ uniformly at random, and let a be the number of them that belong to $A$. Then $\hat{p}=\frac{a}{m}$ is an unbiased estimator of $p$, and it suffices $m=poly(p^{-1},\epsilon^{-1},\log \delta^{-1})$ in order to have 
\[\Pr[ (1-\epsilon) p \leq  \hat{p}\leq (1+\epsilon) p] \geq 1-\delta.\]   
\end{lemma}

If $f(x)\neq 0$, then $p>\frac{1}{2},$ so by the unbiased estimator of lemma \ref{unbiased estimator}, $\widehat{f(x)}$ satisfies the definition of fpras. If $f(x)=0$ then $\widehat{f(x)}=0$, so the estimated value is $0$ with probability~1.

For the second inclusion, let $f\in \FPRASP$,\ we will show that the decision version of $f$, i.e. deciding if $f(x)=0$, is in \RP.\ On input $x$ we run the fpras for $f$ with e.g. $\epsilon=\delta=\frac{1}{4}.$ We return yes iff $\widehat{f(x)}\geq \frac{1}{2}.$ 

By the definition of $\FPRASP$, if $f(x)=0$ then the fpras returns $0$, so we return yes with probability $0$. If $f(x)\geq 1$, then $\widehat{f(x)}\geq \frac{1}{2}$ with probability at least $1-\delta$, so we return yes with the same probability.
\end{proof}
\medskip

\begin{corollary}\label{Corollary1.5}
$\RPo\subseteq \FPRASP\subseteq\FPRAS\subseteq\sBPP$.
\end{corollary}

\begin{corollary}\label{Corollary2}
If $\FPRAS \subseteq \totp$ then \RP=\cP.
\end{corollary}
\begin{proof}
If \FPRAS\ $\subseteq$ \totp,\ then \RPo\ $\subseteq$ \totp,\ and then for all $f \in \RPo$, $L_f\in$ \cP. So if $A\in \RP$ via $M\in \mathcal{MR}$ then $\#acc_M\in$ \RPo, and thus $A=L_{\#acc_M}\in$ \cP. Thus \RP=\cP.
\end{proof}
\medskip

\begin{corollary}\label{Corollary3}
If $\RPo=\RPt$\ then\ \NP=\RP.
\end{corollary}
\begin{proof} 
If \RPo=\RPt\ then they are both equal to \FPRASP,\ thus \totp\ $\subseteq$ \FPRASP$\subseteq$ \FPRAS. Therefore, \NP=\RP\ by Corollary \ref{Corollary1}.
\end{proof}

\medskip

Theorems~\ref{line of inclusions} and~\ref{main theorem} together with Theorem~\ref{PZmt} are summarised in Figure~\ref{fig1}.

\subsection{Conditional inclusions / Possible worlds}

Now we will explore further relationships between the above mentioned classes, and we will present two possible worlds inside \sP, with respect to \NP\ vs. \RP\ vs. \cP.

\begin{theorem}\label{four worlds}
The inclusions depicted in Figure~\ref{fig2} hold under the corresponding assumptions on top of each subfigure. 
\end{theorem}
\begin{proof} 
First note that intersections between any of the above classes are non-empty, because \FP\ is a subclass of all of them.
For the rest of the inclusions, we have the following.
\begin{itemize}
\item In the case of $\NP\neq\RP=\cP$.
\begin{itemize}
\item \vspace{3mm} By definitions, \sP\ $\subseteq$\RPt\ $\Leftrightarrow $ \NP=\RP. Therefore, $$\NP\neq\RP\Rightarrow \sP\not\subseteq\RPt.$$
\item By Theorem~\ref{PZmt}, the inclusions \FP\ $\subseteq$ \totp\ $\subseteq \sPE\ \subseteq \sP$ are proper unless $\cP=\NP$. Therefore,
$$\NP\neq\cP\Rightarrow \FP \subsetneq \totp \subsetneq \sPE\subsetneq\sP.$$
\item By Corollary~\ref{Corollary1}, \totp\ $\subseteq$ \FPRAS\ $\Rightarrow$ \NP=\RP. Therefore,
$$\NP\neq\RP\Rightarrow \totp\not\subseteq\FPRAS.$$
\item By Corollary~\ref{Corollary1}, \totp\ $\subseteq$ \FPRASP\ $\Rightarrow$ \NP=\RP. Therefore,
$$\NP\neq\RP\Rightarrow \totp\not\subseteq\FPRASP.$$
\item By Corollary~\ref{Corollary1} and Theorem~\ref{main theorem},
\RPt\ $\subseteq$ \FPRAS\ $\Rightarrow$ \totp\ $\subseteq$ \FPRAS\ $\Rightarrow$ \NP=\RP. Therefore,
$$\NP\neq\RP\Rightarrow\RPt\not\subseteq\FPRAS.$$
\item By Theorem~\ref{main theorem} and Corollary~\ref{Corollary1},
\RPt\ $\subseteq$ \FPRASP\ $\Rightarrow$ \totp\ $\subseteq$ \FPRASP\ $\Rightarrow$ \NP=\RP. Therefore,
$$\NP\neq\RP\Rightarrow\RPt\not\subseteq\FPRASP.$$
\item By Corollary~\ref{Corollary3}, \RPt=\RPo\ $\Rightarrow$ \NP=\RP. Therefore,
$$\NP\neq\RP\Rightarrow\RPt\not\subseteq\RPo.$$
\item By Theorem~\ref{sP vs fpras},
\sP\ $\subseteq$ \FPRAS\ $\Leftrightarrow$ \NP=\RP. Therefore,
$$\NP\neq\RP\Rightarrow\sP\not\subseteq\FPRAS.$$
\item By Theorem~\ref{PZmt} and Corollary~\ref{Corollary1},
\sPE\ $\subseteq$ \FPRAS\ $\Rightarrow$ \totp\ $\subseteq$ \FPRAS\ $\Rightarrow$ \NP=\RP. Therefore,
$$\NP\neq\RP\Rightarrow\sPE\not\subseteq\FPRAS.$$
\item By Theorem~\ref{main theorem} and the previous result,
\sPE\ $\subseteq$ \RPo\ $\Rightarrow$ \sPE\ $\subseteq$ \FPRAS\  $\Rightarrow$ \NP=\RP. Therefore,
$$\NP\neq\RP\Rightarrow\sPE\not\subseteq\RPo.$$
\item  By Theorem~\ref{PZmt} and Corollary~\ref{Corollary1},
\sPE\ $\subseteq$ \FPRASP\ $\Rightarrow$ \totp\ $\subseteq$ \FPRASP\ $\Rightarrow$ \NP=\RP. Therefore,
$$\NP\neq\RP\Rightarrow\sPE\not\subseteq\FPRASP.$$
\item By Corollary~\ref{Corollary1} and Theorem~\ref{main theorem},
\totp\ $\subseteq$ \RPo $\Rightarrow$ \totp\ $\subseteq$ \FPRAS\ $\Rightarrow$ \NP=\RP. Therefore,
$$\NP\neq\RP\Rightarrow\totp\not\subseteq\RPo.$$
\end{itemize}
\item \vspace{3mm}In the case of $\NP\neq\RP\neq\cP$.\\
In addition to all the above results we have the following ones.
\begin{itemize}
\item \vspace{2mm} By definitions, \RPt\ $\subseteq$ \sPE\ $\Leftrightarrow$ \cP=\RP. Therefore, $$\cP\neq\RP\Rightarrow \RPt\not\subseteq\sPE.$$
\item As in the proof of Corollary~\ref{Corollary2} we can show that
\RPo\ $\subseteq$ \sPE\ $\Rightarrow$ \cP=\RP\ holds. Therefore,
$$\cP\neq\RP\Rightarrow\RPo\not\subseteq\sPE.$$
\item By Theorem~\ref{main theorem} and the previous result,
\FPRAS\ $\subseteq$ \sPE\ $\Rightarrow$ \RPo\ $\subseteq$ \sPE\ $\Rightarrow$ \cP=\RP. Therefore,
$$\cP\neq\RP\Rightarrow\FPRAS\not\subseteq\sPE.$$
\item Similarly, \FPRASP\ $\subseteq$ \sPE\ $\Rightarrow$ \RPo\ $\subseteq$ \sPE\ $\Rightarrow$ \cP=\RP. Therefore,
$$\cP\neq\RP\Rightarrow\FPRASP\not\subseteq\sPE.$$
\item Similarly, 
\RPo\ $\subseteq$ \totp\ $\Rightarrow$ \cP=\RP. Therefore,
$$\cP\neq\RP\Rightarrow\RPo\not\subseteq\totp.$$
\item By Theorem~\ref{PZmt} and the previous result,
\RPo\ $\subseteq$ \FP\ $\Rightarrow$ \RPo\ $\subseteq$ \totp\ $\Rightarrow$ \cP=\RP. Therefore,
$$\cP\neq\RP\Rightarrow\RPo\not\subseteq\FP.$$
\item By Corollary~\ref{Corollary2},
\FPRAS\ $\subseteq$ \totp\ $\Rightarrow$ \cP=\RP. Therefore,
$$\cP\neq\RP\Rightarrow\FPRAS\not\subseteq \totp.$$
\item Similarly, \FPRASP\ $\subseteq$ \totp\ $\Rightarrow$ \cP=\RP. Therefore,
$$\cP\neq\RP\Rightarrow\FPRASP\not\subseteq \totp.$$
\end{itemize}
\end{itemize}
\end{proof}

\section{Conclusions and open questions}\label{discussion}

\begin{wrapfigure}{r} {0.5\textwidth}
	\vspace{-20pt}
	\centering
\begin{tikzpicture}
\node at (1,3) {\sP};
\draw [thick, |-> ] (1.7,2.3) -- (1.3,2.7);
\node at (-0.5,2) {\sPE=\RPt};
\draw [thick, -|] (0.7,2) -- (1.2,2);
\draw [thick, |->] (0,2.4) -- (0.6,2.8);
\node at (0,1) {\totp};
\draw [thick, -|] (0.7,1.3) -- (1.4,1.7);
\draw [thick, |->] (0,1.3) -- (0,1.7);
\draw [thick, -|] (0.7,1) -- (2.4,1);
\node at (2,2) {\FPRAS};
\node at (3,1) {\RPo};
\draw [thick, ->] (2.7,1.2) -- (2.3,1.6);
\draw [thick, -|] (2,0.3) -- (2,1.7);
\node at (2,0) {$\mathsf{\Sigma QSO(\Sigma_2\text{-}2SAT)}$};
\draw [thick, ->] (1.3,0.3) -- (0.7,0.7);
\node at (-1,0) {$\#\mathsf{\Pi_2\text{-}1VAR}$};
\draw [thick, ->] (-0.6,0.3) -- (-0.1,0.7);
\draw [thick, -|] (-0.1,0.4) -- (1.7,1.5);
\end{tikzpicture}
\caption{Inclusions and separations in the case of $\NP\neq\RP=\cP$.}

	\vspace{-10pt}
	\label{fig4}
\end{wrapfigure}

Regarding the question of whether \FPRAS\ is a subset of \totp, Corollary~\ref{Corollary2} states that if it actually  holds, then  proving it is at least as difficult as proving \RP=\cP. 
 
A long-sought structural characterization for \FPRAS\ might be obtained by exploring the fact that it lies between \RPo\ and \sBPP. 

Another open question is whether \FPRASP\ is included in \RPo. It seems that both a negative and a positive answer are compatible with our two possible worlds. 

\begin{figure}[ht]
\centering
\vfill\centering
\begin{tikzpicture}
\node at (1,5) {\sP};
\node at (-1.2,5) {\sPE=\RPt};
\draw [thick, -|] (-0.7,4.6) -- (0.7,3.5);
\draw [thick, |-> ] (-0.1,5) -- (0.6,5);
\node at (1,3) {\FPRAS};
\draw [thick, |->] (1,3.4) -- (1,4.7);
\node at (-1,3) {\totp};
\draw [thick, -|] (-0.5,3) -- (0.3,3);
\draw [thick, |->] (-1,3.4) -- (-1,4.6);
\node at (-4,3) {$\mathsf{\Sigma QSO(\Sigma_2\text{-}HORN)}$};
\draw [thick, ->] (-2.5,3) -- (-1.5,3);
\node at (-2.5,1) {$\mathsf{spanL}$};
\draw [thick, ->] (-2.7,1.3) -- (-1.3,2.5);
\draw [thick, ->] (-2,1.3) -- (0.8,2.5);
\node at (0,1) {$\mathsf{\#R\Sigma_2}$};
\draw [thick, ->] (-0.3,1.3) -- (-0.8,2.5);
\draw [thick, ->] (0.3,1.3) -- (1,2.5);
\node at (3,4) {$\mathsf{\Sigma QSO(\Sigma_2\text{-}2SAT)}$};
\draw [thick, ->] (1.5,4) -- (-0.5,3.4);
\draw [thick, -|] (1.8,3.7) -- (1.5,3.3);
\node at (3,1) {$\mathsf{\Sigma QSO(\Sigma_1[FO])}$};
\draw [thick, ->] (2.3,1.5) -- (-0.3,2.5);
\draw [thick, ->] (3,1.5) -- (1.8,2.5);
\node at (-2.5,0) {\FP};
\draw [thick, ->] (-2.5,0.3) -- (-2.5,0.7);
\node at (-2.5,-1) {$\#\mathsf{\Sigma_0}$};
\draw [thick, |->] (-1.5,-1) -- (2,-1);
\draw [thick, ->] (-2.5,-0.7) -- (-2.5,-0.3);
\node at (3,0) {$\mathsf{\Sigma QSO(\Sigma_1)}$};
\draw [thick, ->] (3,0.3) -- (3,0.7);
\node at (3,-1) {$\#\mathsf{\Sigma_1}$};
\draw [thick, |->] (3,-0.7) -- (3,-0.3);
\node at (3.6,3) {$\mathsf{\#RH\Pi_1}$};
\draw [thick, ->] (3.6,3.3) -- (3.6,3.7);
\end{tikzpicture} 
\vspace{2mm}
\caption{Inclusions and separations in the case of $\NP\neq\RP=\cP$.}
\label{fig5}
\end{figure}

By employing descriptive complexity methods we obtained two new robust subclasses of \totp; the class $\mathsf{\Sigma QSO(\Sigma_2\text{-}2SAT)}$ for which the counting problem \#\textsc{Disj2Sat} is complete under parsimonious reductions  and the class $\#\mathsf{\Pi_2\text{-}1VAR}$ for which \textsc{\#MonotoneSat} is complete under product reductions. We do not expect $\mathsf{\Sigma QSO(\Sigma_2\text{-}2SAT)}$ to be a subclass of \FPRAS, given that 
\#\textsc{Disj2Sat} does not admit an fpras unless $\NP=\RP$. 

A similar fact holds for the second  class $\#\mathsf{\Pi_2\text{-}1VAR}$. Since there is no fpras for \textsc{\#MonotoneSat} if a variable can appear in $6$ clauses, unless $\NP=\RP$~\cite{LL15}, we do not expect that $\#\mathsf{\Pi_2\text{-}1VAR}$ is a subclass of \FPRAS.

Although proving \textsc{\#MonotoneSat} complete for $\#\mathsf{\Pi_2\text{-}1VAR}$ under product reductions, allows a more precise classification of the problem within \sP, the question of~\cite{HHKW07} remains open, i.e.\  whether \textsc{\#MonotoneSat} is complete for some counting class under reductions under which the class is downwards closed.

Finally, assuming $\NP\neq\RP=\cP$, which is the most widely believed conjecture, the relationships among the classes studied in this paper are given in Figure~\ref{fig4}.

Relationships among \totp, \FPRAS, and various classes defined through descriptive complexity, are shown in Figure~\ref{fig5}.

\medskip

\bibliographystyle{plain}
\bibliography{tamc-42-bibliography}


\end{document}